\documentclass{llncs}
\usepackage{makeidx}
\usepackage{amsfonts}
\usepackage{amsmath}
\usepackage[letterpaper,hmargin=1.2in,vmargin=1in]{geometry}
\usepackage{graphicx}
\usepackage{algorithm}
\usepackage[noend]{algorithmic}
\usepackage[british]{babel}

\newtheorem{observation}{Observation}

\begin{document}

\title{Constructing the Simplest Possible Phylogenetic Network from Triplets\thanks{Part of this research has been funded by the Dutch BSIK/BRICKS project.}}

\author{Leo van Iersel\inst{1} and Steven Kelk\inst{2}}

\authorrunning{Van Iersel, Kelk}

\institute{Department of Mathematics and Computer Science, Technische Universiteit Eindhoven, P.O. Box 513, 5600 MB
Eindhoven, The Netherlands,
\email{l.j.j.v.iersel@tue.nl}\\
\and Centrum voor Wiskunde en Informatica (CWI), P.O. Box 94079, 1090 GB Amsterdam, The Netherlands,
\email{s.m.kelk@cwi.nl} \\
}

\maketitle

\begin{abstract}
A phylogenetic network is a directed acyclic graph that visualises
an evolutionary history containing so-called \emph{reticulations}
such as recombinations, hybridisations or lateral gene transfers.
Here we consider the construction of a simplest possible
phylogenetic network consistent with an input set $T$, where
$T$ contains at least one phylogenetic tree on three leaves (a \emph{triplet}) for
each combination of three taxa. To quantify the complexity of a
network we consider both the total number of reticulations and the
number of reticulations per biconnected component, called the
\emph{level} of the network. We give polynomial-time algorithms for
constructing a level-1 respectively a level-2 network that contains
a minimum number of reticulations and is consistent with $T$ (if
such a network exists). In addition, we show that if $T$ is
precisely equal to the set of triplets consistent with some network,
then we can construct such a network with smallest possible level in
time $O(|T|^{k+1})$, if $k$ is a fixed upper bound on the level of
the network.
\end{abstract}

\section{Introduction}
One of the ultimate goals in computational biology is to create methods that can reconstruct evolutionary histories
from biological data of currently living organisms. The immense complexity of biological evolution makes this task
almost a hopeless one \cite{morrison}. This has motivated researchers to focus first on the simplest possible pattern
of evolution. This least complicated shape of an evolutionary history is the tree-shape. Now that treelike evolution has
been extremely well studied, a logical next step is to consider slightly more complicated evolutionary scenarios,
gradually extending the complexity that our models can describe. At the same time we also wish to take into account the
parsimony principle, which tells us that amongst all equally good explanations of our data, one prefers the simplest one
(see e.g. \cite{hein}).\\
\\
For a set of taxa (e.g. species or strains), a phylogenetic tree describes (a hypothesis of) the evolution that these
taxa have undergone. The taxa form the leaves of the tree while the internal vertices represent events of genetic
divergence: one incoming branch splits into two (or more) outgoing branches.\\
\\
Phylogenetic networks form an extension to this model where it is also possible that two branches combine into one new
branch. We call such an event a \emph{reticulation}, which can model any kind of non-treelike (also called
``reticulate'') evolutionary process such as recombination, hybridisation or lateral gene transfer. In addition,
reticulations can also be used to display different possible (treelike) evolutions in one figure. In recent years
there has emerged enormous interest in phylogenetic networks and their application
\cite{baroni2}\cite{husonbryant}\cite{makarenkov}\cite{morrison}\cite{moret2004}.\\
\\
This model of a phylogenetic network allows for many different degrees of complexity, ranging from networks that are
equal, or almost equal, to a tree to unrealistically complex webs of frequently diverging and recombining lineages.
Therefore we consider two different measures for the complexity of a network. The first of these measures is the total
number of reticulations in the network. Secondly, we consider the \emph{level} of the network, which is an upper bound
on the number of reticulations per biconnected component of the network. Informally, the level of a network is a bound
on the number of reticulations that can be mutually dependent. In this paper we consider two different approaches for
constructing networks that are as simple as possible. The first approach minimises the total number of reticulations
for a fixed level (of at most two) and the second approach minimises (under certain special restrictions on the input)
the level while the total number of reticulations is unrestricted.\\
\\
Level-$k$ phylogenetic networks were first introduced by Choy et al. \cite{choy} and further studied by different
authors \cite{recomb}\cite{reflections}\cite{JS1}. Gusfield et al. gave a biological justification for level-1
networks (which they call ``galled trees'') \cite{gusfield}. Minimising reticulations has been very well studied in
the framework where the input consists of (binary) sequences \cite{gusfield2}\cite{hein}\cite{song1}\cite{song2}. For
example, Wang et al. considered the problem of finding a ``perfect phylogeny'' with a minimum number of reticulations
and showed that this problem is NP-hard \cite{wang}. Gusfield et al. showed that this problem can be solved in
polynomial time if restricted to level-1 networks \cite{gusfield}.\\
\\
There are also several results known already about the version of the problem where the input consists of a set of
trees and the objective is to construct a network that is ``consistent'' with each of the input trees. Baroni et al.
give bounds on the minimum number or reticulations needed to combine two trees \cite{baroni} and Bordewich et al.
showed that it is APX-hard to compute this minimum number exactly \cite{bordewich}. However, there exists an exact
algorithm \cite{bordewich2} that runs reasonably fast in many practical situations.\\
\\
In this paper we also consider input sets consisting of trees, but restrict ourselves to small trees with three leaves
each, called \emph{triplets}. See Figure~\ref{fig:singletriplet} for an example. Triplets can for example be
constructed by existing methods, such as Maximum Parsimony or Maximum Likelihood, that work accurately and fast for
small numbers of taxa. Triplet-based methods have also been well-studied. Aho et al. \cite{aho} gave a polynomial-time
algorithm to construct a tree from triplets if there exists a tree that is consistent with all input triplets. Jansson
et al. \cite{JS2} showed that the same is possible for level-1 networks if the input triplet set is \emph{dense}, i.e.
if there is a triplet for any set of three taxa. Van Iersel et al. further extended this result to level-2 networks
\cite{recomb}. From non-dense triplet sets it is NP-hard to construct level-$k$ networks for any $k\geq 1$
\cite{reflections}\cite{JS2}. From the proof of this result also follows directly that it is NP-hard to find a network
consistent with a non-dense triplet set that contains a minimum number of reticulations.\footnote{This follows from
the proof of Theorem~7 in \cite{JS2}, since only one reticulation is used in their reduction.} It is unknown whether
this problem becomes easier if the input triplet set is dense.

\begin{figure}
\centering \vspace{-.2cm} \includegraphics[scale=0.8]{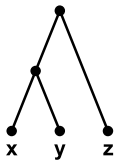} \caption{One of the three possible triplets
on the leaves $x$, $y$ and $z$. Note that, as with all figures in this article, all arcs are directed downwards.}
\vspace{-0.5cm} \label{fig:singletriplet}
\end{figure}

\noindent In the first part of this paper we consider fixed-level networks and aim to minimise the total number of
reticulations in these networks. In Section~\ref{sec:lev1} we give a polynomial-time algorithm that constructs a
level-1 network consistent with a dense triplet set $T$ (if such a network exists) and minimises the total number of
reticulations over all such networks. This gives an extension to the algorithm by Jansson et al. \cite{JS2}, which can
also reconstruct level-1 networks, but does not minimise the number of reticulations. To illustrate this we give in
Section~\ref{sec:prelim} an example dense triplet set on $n$ leaves for which the algorithm in \cite{JS2} (and the
ones in \cite{recomb} and \cite{JS1}) creates a level-1 network with $\frac{n-1}{2}$ reticulations. However, a level-1
network with just one reticulation is also possible and our algorithm MARLON is able to find that network. We have
implemented MARLON (made publicly available \cite{marlon}) and tested it on simulated data. Results are in
Section~\ref{sec:experiments}. The worst case running time of the algorithm is $O(n^5)$ for $n$ leaves (and hence
$O(|T|^{\frac{5}{3}})$ with $|T|$ the input size).\\
\\
In Section~\ref{sec:lev2} we further extend this approach by giving an algorithm that even constructs a level-2
network consistent with a dense triplet set (if one exists) and again minimises the total number of reticulations over
all such networks. This means that if the level is at most two, we can minimise both the level and the total number of
reticulations, giving priority to the criterium that we find most important. The running time is $O(n^9)$ (and
thus $O(|T|^3)$).\\
\\
Minimising the level of phylogenetic networks becomes even more challenging when this level can be larger than two,
even without minimising the total number of reticulations. Given a dense set of triplets, it is a major open problem
whether one can construct a minimum level phylogenetic network consistent with these triplets in polynomial time.
Moreover, it is not even known whether it is possible to construct a level-3 network consistent with a dense input
triplet set in polynomial time. In Section~\ref{sec:extreme} of this paper we show some significant progress in this
direction. As a first step we consider the restriction to ``simple'' networks, i.e. networks that contain just one
nontrivial biconnected component. We show how to construct, in $O(|T|^{k+1})$ time, a minimum level simple network
with level at most $k$ from a dense input triplet set (for fixed $k$). Subsequently we show that this can be used to
also generate general level-$k$ networks if we put an extra restriction on the quality of the input triplets. Namely,
we assume that the input set contains exactly all triplets consistent with some network. If that is the case then our
algorithm can find such a network with a smallest possible level. The algorithm runs in polynomial time $O(|T|^{k+1})$
if the upper bound $k$ on the level of the network is fixed. This result constitutes an important step forward in the
analysis of level-$k$ networks, since it provides the first positive result that can be used for all levels $k$.

\section{Preliminaries}\label{sec:prelim}
A \emph{phylogenetic network} (\emph{network} for short) is defined as a directed acyclic graph in which exactly one
vertex has indegree 0 and outdegree 2 (the root) and all other vertices have either indegree 1 and outdegree 2
(\emph{split vertices}), indegree 2 and outdegree 1 (\emph{reticulation vertices}, or \emph{reticulations} for short)
or indegree 1 and outdegree 0 (\emph{leaves}), where the leaves are distinctly labelled. A phylogenetic network
without reticulations is called a \emph{phylogenetic tree}.\\
\\
\noindent A directed acyclic graph is \emph{connected} (also called ``weakly connected'') if there is an undirected
path between any two vertices and \emph{biconnected} if it contains no vertex whose removal disconnects the graph. A
\emph{biconnected component} of a network is a maximal biconnected subgraph and is called \emph{trivial} if
it is equal to two vertices connected by an arc. To avoid ``redundant'' networks we assume that in any network
every nontrival biconnected component has at least three outgoing arcs. We call an arc $a=(u,v)$ of a network $N$
a \emph{cut-arc} if its removal disconnects $N$ and call it \emph{trivial} if $v$ is a leaf.
\begin{definition}
A network is said to be a \emph{level-$k$ network} if each biconnected component contains at most $k$ reticulations.
\end{definition}
A level-$k$ network that contains no nontrivial cut-arcs and is not a level-$(k-1)$ network is called a \emph{simple}
level-$k$ network\footnote{This definition is equivalent to Definition~4 in \cite{recomb} by Lemma~2 in
\cite{recomb}.}. Informally, a simple network thus consists of a nontrivial biconnected component with leaves
``hanging'' of it.\\
\\
A \emph{triplet} $xy|z$ is a phylogenetic tree on the leaves $x$, $y$ and $z$ such that the lowest common ancestor of
$x$ and $y$ is a proper descendant of the lowest common ancestor of $x$ and $z$. The triplet $xy|z$ is displayed in
Figure~\ref{fig:singletriplet}. Denote the set of leaves in a network $N$ by $L_N$. For any set $T$ of triplets define
$L(T) = \bigcup_{t \in T} L_t$ and let $n=|L(T)|$. A set $T$ of triplets is called \emph{dense} if for each $\{x,y,z\}
\subseteq L(T)$ at least one of $xy|z$, $xz|y$ and $yz|x$ belongs to $T$.\\
\\
For a set of triplets $T$ and a set of leaves $L'\subseteq L(T)$, we denote by $T|L'$ the triplets $t\in T$ with $L_t
\subseteq L'$. Furthermore, if $\mathcal{C}=\{S_1,\ldots,S_q\}$ is a collection of leaf-sets we use
$T\nabla\mathcal{C}$ to denote the \emph{induced} set of triplets $S_i S_j |S_k$ such that there exist $x\in S_i$,
$y\in S_j$, $z\in S_k$ with $xy|z\in T$ and $i$, $j$ and $k$ all distinct.

\begin{definition}
\label{def:con} A triplet $xy|z$ is \emph{consistent} with a network $N$ (interchangeably: $N$ is consistent with
$xy|z$) if $N$ contains a subdivision of $xy|z$, i.e. if $N$ contains vertices $u \neq v$ and pairwise internally
vertex-disjoint paths $u \rightarrow x$, $u \rightarrow y$, $v \rightarrow u$ and $v \rightarrow z$.
\end{definition}

\noindent The above definitions enable us to give a formal description of the problems we consider.\\

\noindent\textbf{Problem:} Minimum Reticulation Level-$k$ network on dense triplet sets (DMRL-$k$).\\
\textbf{Input:} dense set of triplets $T$.\\
\textbf{Output:} level-$k$ network $N$ that is consistent with $T$ (if such a network exists) and has a minimum number
of reticulations over all such networks.\\
\\
A feasible solution to DMRL-1 can be found by the algorithm in \cite{JS2} or \cite{JS1} and the algorithm in
\cite{recomb} finds a feasible solution to DMRL-2. To show that these algorithms do not always minimise the number of
reticulations, consider a triplet set over an odd number $n$ of leaves, labelled $1,\ldots,n$, containing all triplets
$ab|c$ with $a,b>c$ and the triplets $a(a+1)|n$ with $a=1,3,\ldots,n-2$. The aforementioned algorithms find for this
input set a level-1 network with $\frac{n-1}{2}$ reticulations. However, a level-1 network with just one reticulation
is also possible and our algorithm MARLON, introduced shortly, is able to find that network. See Figure~\ref{fig:counterex} for an example
for $n=9$.

\begin{figure}[h]\begin{minipage}[b]{0.5\textwidth}
{\centerline{\includegraphics[scale=.3]{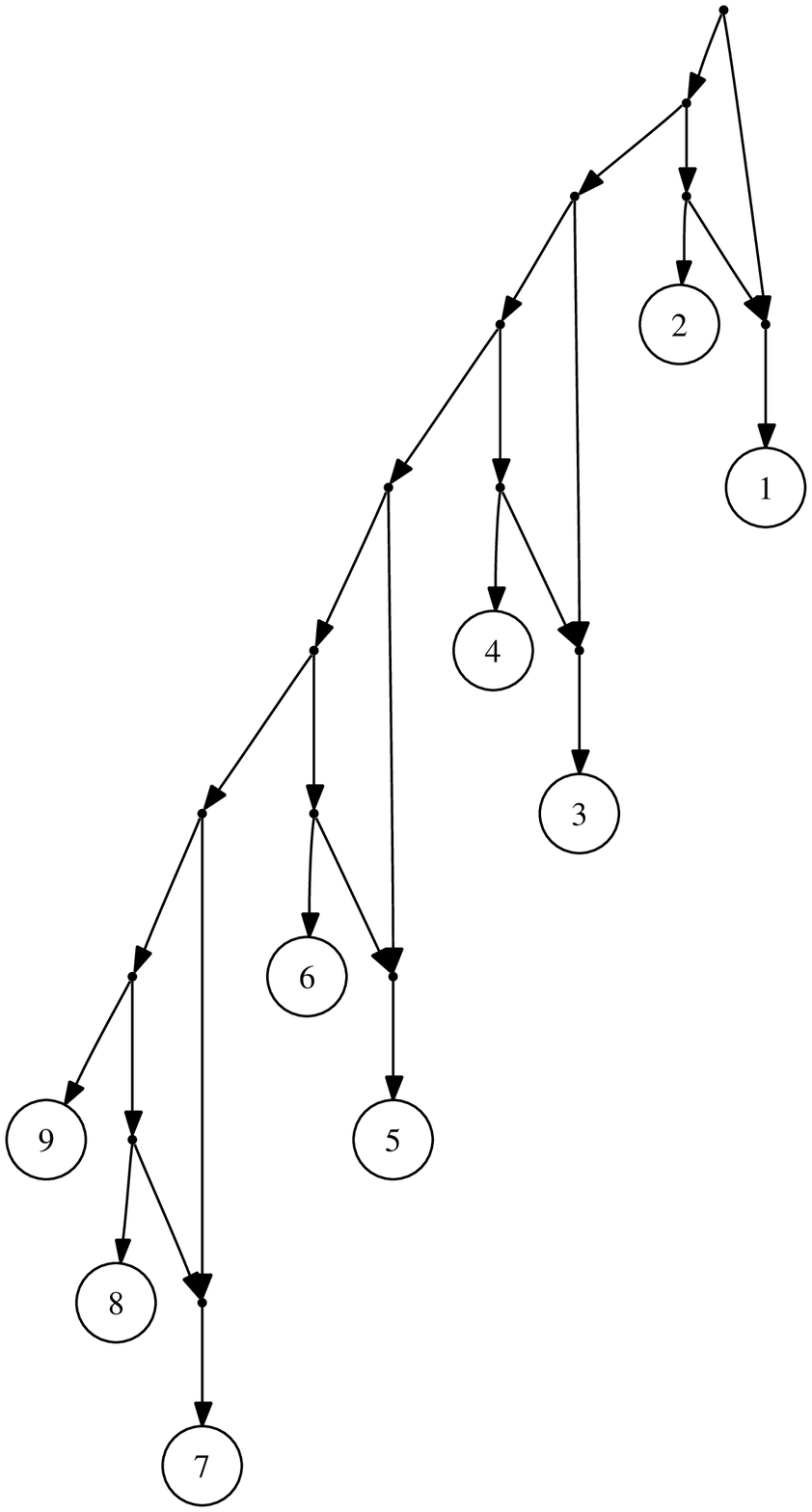}}}\end{minipage}
\begin{minipage}[b]{0.5\textwidth}
{\centerline{\includegraphics[scale=.3]{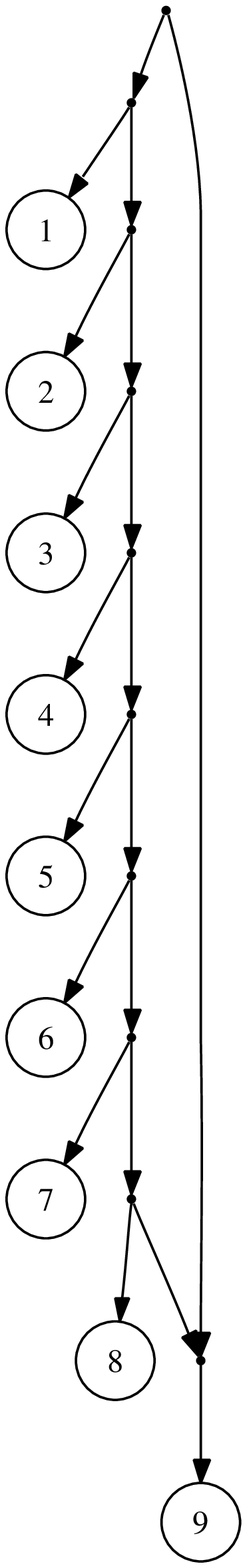}}}\end{minipage} \caption{Example of a situation where previous
algorithms (by Jansson et al. \cite{JS2} and Van Iersel et al. \cite{recomb}) construct a network like the one to the
left with $\frac{n-1}{2}$ reticulations, while MARLON constructs the network to the right, with just one
reticulation.}\label{fig:counterex}\vspace{-.3cm}
\end{figure}

\noindent Given a network $N$ let $T(N)$ denote the set of all triplets consistent with $N$. We say that a
network $N$ \emph{reflects} a triplet set $T$ if $T(N)=T$. If, for a triplet set $T$, there exists a network $N$
that reflects it, we say that $T$ is \emph{reflective}. The second problem we consider is thus the following:\\
\\
\noindent\textbf{Problem:} MIN-REFLECT-$k$\\
\textbf{Input:} set of triplets $T$.\\
\textbf{Output:} level-$k$ network $N$ that reflects $T$ (if such a network exists) and has the smallest possible
level over all such networks.\\

\noindent Note that this problem is closely related to the mixed triplets problem (MT) studied in \cite{forbid}, which
asks for a phylogenetic network consistent with an input triplet set $T$ and \emph{not} consistent with another input
triplet set $F$. Namely, MIN-REFLECT-$k$ is a special case of MT restricted to level-$k$ networks where the set $F$ of
forbidden triplets contains all triplets that are not in $T$.\\

\noindent To describe our algorithms we need to introduce some more definitions. We say that a cut-arc is a
\emph{highest cut-arc} if it is not reachable from any other cut-arc. We call a cycle containing the root a
\emph{highest cycle} and a reticulation in such a cycle a \emph{highest reticulation}. We say that a leaf $x$ is
\emph{below} an arc $(u,v)$ (and \emph{below} vertex $v$) if $x$ is reachable from $v$. In the next section we will
frequently use the set $BHR(N)$, which denotes the set of leaves in network $N$ that is below a highest
reticulation.\\
\\
A subset $S$ of the leaves is an \emph{SN-set} (w.r.t. triplet set $T$) if there is no triplet $xy|z$ in $T$ with
$x,z\in S$, $y\notin S$. An SN-set is called \emph{nontrivial} if it does not contain all leaves. We say that an
SN-set $S$ is \emph{maximal} (under restriction $X$) if there is no nontrivial SN-set (satisfying restriction $X$)
that is a strict superset of $S$. An SN-set of size 1 is called a \emph{singleton} SN-set.\\
\\
Any two SN-sets w.r.t. a dense triplet set are either disjoint or one is included in the other \cite[Lemma~8]{JS1},
which leads to the following definition. The \emph{SN-tree} is a directed tree with vertices with outdegree greater or
equal to two, such that the SN-sets of $T$ correspond exactly to the sets of leaves reachable from a vertex of the
SN-tree. It follows that there are at most $2(n-1)$ nontrivial SN-sets in a dense triplet set $T$. All these SN-sets
can be found by constructing the SN-tree in $O(n^3)$ time \cite{JS2}. If a network is consistent with a dense triplet
set $T$, then the set of leaves $S$ below any cut-arc is always an SN-set, since triplets of the form $xy|z$ with
$x,z\in S$, $y\notin S$, are not consistent with such a network. Furthermore, each maximal SN-set is equal to the
union of leaves below one or more highest cut-arcs \cite[Lemma~5]{arxiv}.

\section{Constructing a Level-1 Network with a Minimum Number of Reticulations}\label{sec:lev1}

We propose the following dynamic programming algorithm for solving DMRL-1. The algorithm considers all SN-sets from
small to large and computes an optimal solution $N_S$ for each SN-set $S$, based on the optimal solutions for included
SN-sets. The algorithm considers both the case where the root of $N_S$ is contained in a cycle and the case where
there are two cut-arcs leaving the root. In the latter case there are two SN-sets $S_1$ and $S_2$ that are maximal
under the restriction that they are a subset of $S$. If this is the case then the algorithm constructs a candidate for
$N_S$ by creating a root connected to the roots of $N_{S_1}$ and $N_{S_2}$.\\
\\
The other possibility is that the root of $N_S$ is contained in some cycle. For this case the algorithm tries each
SN-set as $BHR(N_S)$: the set of leaves below the highest reticulation. The sets of leaves below other highest
cut-arcs can then be found using the property of optimal level-1 networks outlined in Lemma~\ref{lem:snsets}.
Subsequently, an induced set of triplets is computed, where each set of leaves below a highest cut-arc is replaced by
a single meta-leaf. A candidate network is constructed by computing a simple level-1 network and replacing each
meta-leaf $S_i$ by an optimal network $N_{S_i}$ for the corresponding subset of the leaves. The optimal network $N_S$
is then the network with a minimum number of reticulations over all computed networks.\\
\\
A structured description of the computations is in Algorithm~\ref{alg:MARLON}. We use $f(L')$ to denote the minimum
number of reticulations in any level-1 network consistent with $T|L'$. In addition, $g(L',S')$ denotes the minimum
number of reticulations in any level-1 network consistent with $T|L'$ with $BHR(N)=S'$. The algorithm first computes
the optimal number of reticulations. Then a network with this number of reticulations is constructed using
backtracking.\\

\begin{algorithm}[h]
\caption{MARLON (Minimum Amount of Reticulation Level One Network)} \label{alg:MARLON}
\begin{algorithmic} [1]
\STATE compute the set $SN$ of SN-sets w.r.t. $T$\\
\FOR{$i=1\ldots n$} \FOR{each $S$ in $SN$ of cardinality $i$}

\FOR{each $S'\in SN$ with $S'\subset S$}

\STATE let $\mathcal{C}$ contain $S'$ and all SN-sets that are maximal under the restriction that they are a subset of $S$ and do not contain $S'$\\

\IF{$T\nabla \mathcal{C}$ is consistent with a simple level-1 network}

\STATE $g(S,S') := 1 + \sum_{X\in \mathcal{C}} f(X)$\\

\ENDIF

\ENDFOR

\IF{there are exactly two SN-sets $S_1,S_2\in SN$ that are maximal under the restriction that they are a strict subset
of $S$}

\STATE $g(S,\emptyset) := f(S_1) + f(S_2)$ ($\mathcal{C} := \{S_1,S_2\}$)\\

\ENDIF

\STATE $f(S) := \min g(S,S')$ over all computed values of $g(S,\cdot)$\\

\STATE store the optimal $\mathcal{C}$ and the corresponding simple level-1 network\\

\ENDFOR \ENDFOR

\STATE construct an optimal network by backtracking.\\

\end{algorithmic}
\end{algorithm}

\noindent To show that the algorithm indeed computes an optimal solution we need the following crucial property of
optimal level-1 networks.

\begin{lemma}\label{lem:snsets}
If there exists a solution to DMRL-1, then there also exists an optimal solution $N$, where the sets of leaves below
highest cut-arcs equal either (i) $BHR(N)$ and the SN-sets that are maximal under the restriction that they do not
contain $BHR(N)$, or (ii) the maximal SN-sets (if $BHR(N)=\emptyset$).
\end{lemma}
\begin{proof}
If $BHR(N)=\emptyset$ then there are two highest cut-arcs and the sets below them are the maximal SN-sets. Otherwise,
the root of $N$ is part of a cycle. Let $S$ be a maximal SN-set. We prove the following.
\begin{claim}[1]
Maximal SN-set $S$ equals either the set of leaves below a highest cut-arc or the set of leaves below a directed path
$P$ ending in the highest reticulation or in one of its parents.
\end{claim}
\begin{proof}
If $S$ equals the set of leaves below a single highest cut-arc then we are done. From now on assume that $S$ equals
the set of leaves below different highest cut-arcs. First observe that no two leaves in $S$ have the root as their
lowest common ancestor, since this would imply that \emph{all} leaves are in $S$, because $S$ is an SN-set. From this
follows that all leaves in $S$ are below some directed path $P$ on the highest cycle. First assume that not all leaves
reachable from vertices in $P$ are in $S$. Then there are leaves $x,z,y$ reachable respectively from vertices
$p_1,p_2,p_3$ that are on $P$ (in this order) with $x,y\in S$ and $z\notin S$. But this leads to a contradiction
because then the triplet $xy|z$ is not consistent with $N$, whilst $yz|x$ and $xz|y$ cannot be in $T$ since $S$ is an
SN-set. It remains to prove that $P$ ends in either the highest reticulation or in one of its parents. Assume that
this is not true, then there exists a vertex $v$ on (the interior of) a path from the last vertex of $P$ to the
highest reticulation. Consider some leaf $z\notin S$ reachable from $v$ and some leaves $x,y\in S$ below different
highest cut-arcs. Then this again leads to a contradiction because $xy|z$ is not consistent with $N$. This concludes
the proof of the claim. \qed
\end{proof}
First suppose that a maximal SN-set $S$ equals the set of leaves below a directed path $P$ ending in a parent of the
highest reticulation. In this case we can modify the network by putting $S$ below a single cut-arc, without increasing
the number of reticulations. To be precise, if $p$ and $p'$ are the first and last vertex of $P$ respectively and $r$
is the highest reticulation, then we subdivide the arc entering $p$ by a new vertex $v$, add a new arc $(v,r)$, remove
the arc $(p',r)$ and suppress the resulting vertex with indegree and outdegree both equal to one. It is not too
difficult to see that the resulting network is still consistent with $T$.\\

\begin{figure}[h]\centering
{\centerline{\includegraphics[scale=.6]{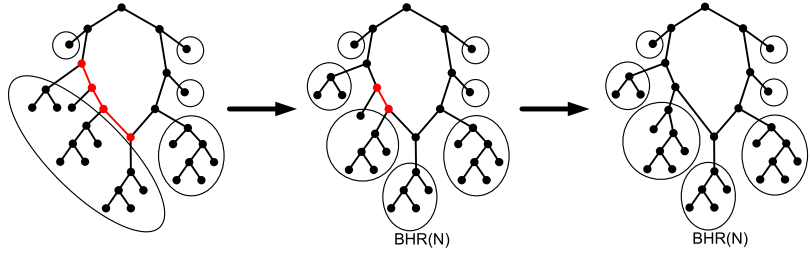}}}\caption{Visualisation of the proof of Lemma~\ref{lem:snsets}. From
the maximal SN-sets (encircled in the network on the left) to the sets of leaves below highest cut-arcs (encircled in
the network on the right). Remember that all arcs are directed downwards.}\label{fig:lev1ex}
\end{figure}

\noindent Now suppose that some maximal SN-set $S$ equals the set of leaves below a directed path $P$ ending in the
highest reticulation. The sets of leaves below highest cut-arcs are all SN-sets (as is always the case). One of them
is equal to $BHR(N)$. If any of the others is contained in a nontrivial SN-set $S'$ that does not contain $BHR(N)$,
then the procedure from the previous paragraph can again be used to put $S'$ below a highest cut-arc. In the resulting
network the sets of leaves below highest cut-arcs are indeed equal to $BHR(N)$ and the SN-sets that
are maximal under the restriction that they do not contain $BHR(N)$.\\
\\
An example is given in Figure~\ref{fig:lev1ex}. In the network on the left one maximal SN-set equals the set of leaves
below the red path. In the middle is the same network, but now we encircled $BHR(N)$ and the SN-sets that are maximal
under the restriction that they do not contain $BHR(N)$. There is still an SN-set ($S'$) below a path on the cycle
(again in red). However, in this case the network can be modified by putting $S'$ below a single cut-arc, without
increasing the number of reticulations. This gives the network to the right, where the sets of leaves below highest
cut-arcs are indeed equal to $BHR(N)$ and the SN-sets that are maximal under the restriction that they do not contain
$BHR(N)$.\qed
\end{proof}

\begin{theorem}\label{thm:lev1}
Given a dense set of triplets $T$, algorithm MARLON constructs a level-1 network that is consistent with $T$ (if such
a network exists) and has a minimum number of reticulations in $O(n^5)$ time.
\end{theorem}
\begin{proof}
The proof is by induction on the size $i$ of $S$. Suppose that $N$ is an optimal level-1 network consistent with
$T|S$. If $BHR(N)=\emptyset$ then the sets of leaves below highest cut-arcs are the two maximal SN-sets $S_1$ and
$S_2$. In this case $f(S)$ can be computed by adding up the $f(S_1)$ and $f(S_2)$. Otherwise, it follows from
Lemma~\ref{lem:snsets} and the observation that $BHR(N)$ has to be an SN-set, that at some iteration the algorithm
will consider the set $\mathcal{C}$ equal to the sets of leaves below the highest cut-arcs of $N$. In this case the
number of reticulations can be computed by adding one to the sum of the values $f(X)$ over all $X\in \mathcal{C}$.
This is because the network $N$ consists of a (highest) cycle, connected to optimal networks for the different $X\in
\mathcal{C}$. By induction, all values of $f(X)$ for $|X|<i$ have been computed correctly and correctness of the
algorithm follows. The number of SN-sets is $O(n)$ because any two SN-sets are either disjoint or one is included in
the other \cite[Lemma~8]{JS1}. These SN-sets can be found in $O(n^3)$ time by computing the SN-tree \cite{JS2}. Simple
level-1 networks can be found in $O(n^3)$ time \cite{JS2} and $T\nabla \mathcal{C}$ can be computed in $O(n^3)$ time.
These computations are repeated $O(n^2)$ times: for all $S\in SN$ and all $S'\in SN$ with $S'\subset S$. Therefore,
the total running time is $O(n^5)$. \qed
\end{proof}

\section{Experiments}\label{sec:experiments}
MARLON has been implemented, tested and made publicly available \cite{marlon}. For example the network in
Figure~\ref{fig:marlon} with 80 leaves and 13 reticulations could be constructed by MARLON in less than six minutes on
a Pentium IV 3 GHz PC with 1 GB of
RAM.\\

\begin{figure}[t]\centering
{\centerline{\includegraphics[scale=.25]{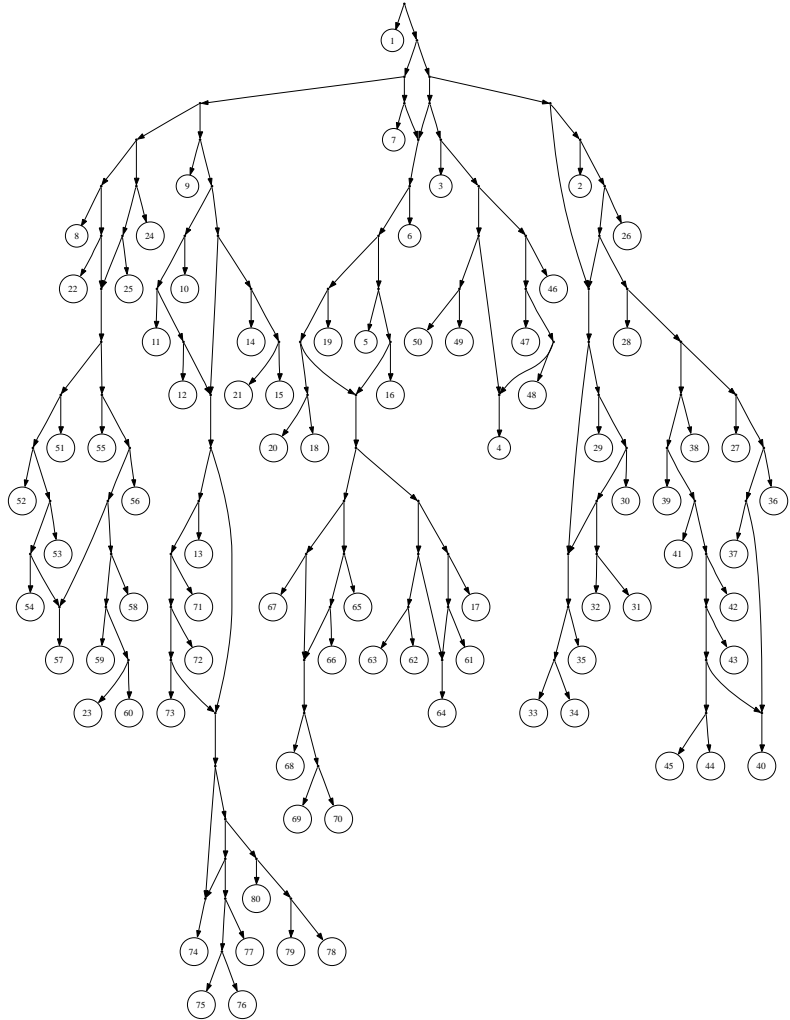}}}\caption{Example of a network constructed by
MARLON.}\label{fig:marlon}
\end{figure}

\noindent To test the relevance of the constructed networks we applied MARLON to simulated data. The main advantage of
using simulated data is that it enables us to compare output networks with the ``real'' network. We repeated the
following experiment for different level-1 networks, which we in turn assumed to be the ``real'' network. Given such a
level-1 network, we used the program Seq-Gen \cite{seqgen} to simulate sequences that could have evolved according to
that network. We simulated a recombinant phylogeny by generating sequences of 4000 base pairs, consisting of two
blocks of 2000 base pairs each. We assumed that each block evolved according to a phylogenetic tree. This means that
in each simulation, our input to Seq-Gen consisted of two trees $T_1$ and $T_2$. For each reticulation of the level-1
network, $T_1$ uses just one of the incoming arcs and $T_2$ uses the other one. This makes sure that each arc of the
network is used by at least one of the two trees. Seq-Gen was used with the F81 model of nucleotide
substitution.\\

\begin{figure}[t]\centering
\vspace{-1cm} {\centerline{\includegraphics[scale=.28]{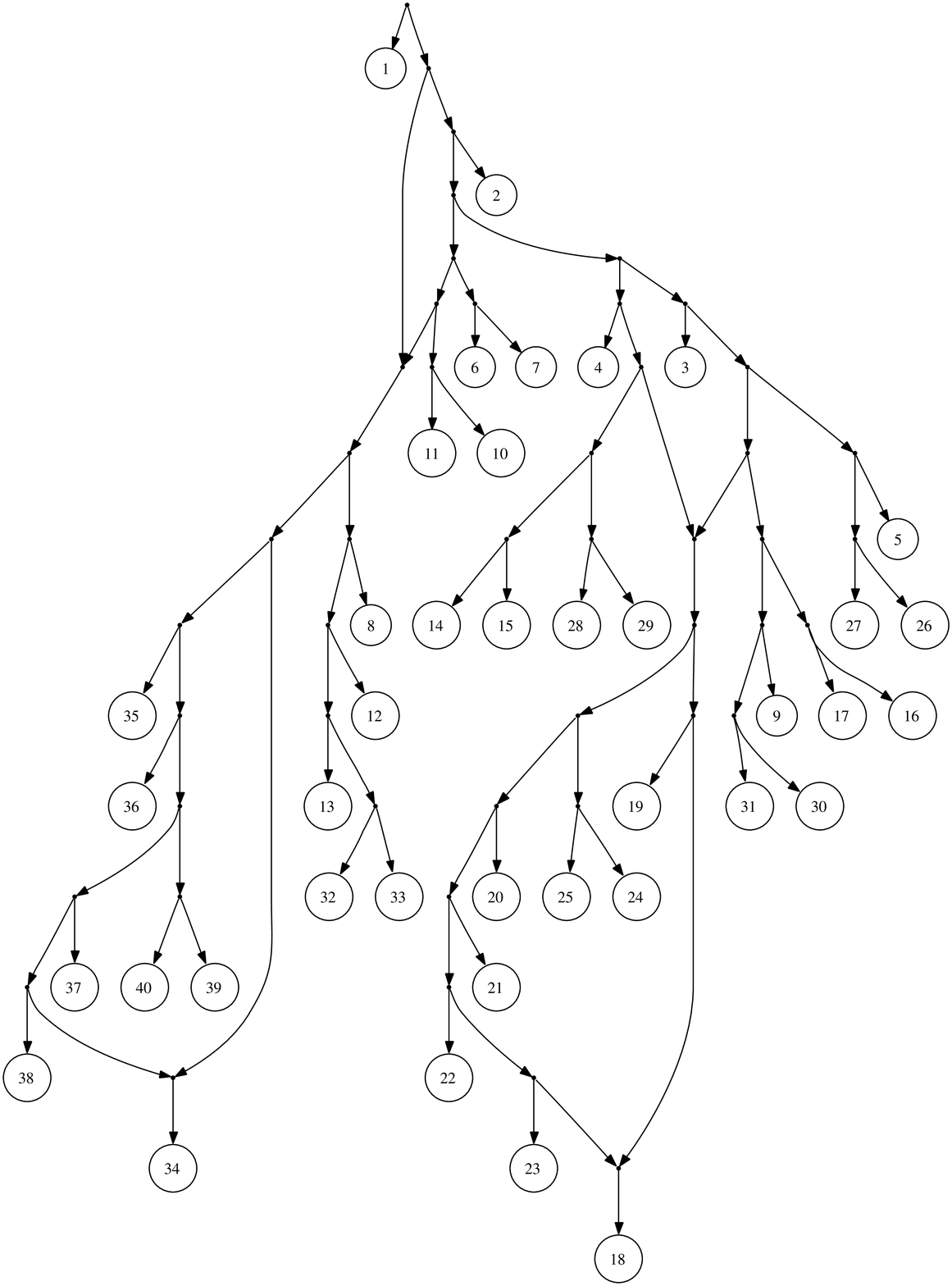}}}\caption{The level-1 network on which the simulated
triplet set $T^*$ is based.}\label{fig:input40}
\end{figure}

\begin{figure}[t]\centering
\vspace{-1cm} {\centerline{\includegraphics[scale=.28]{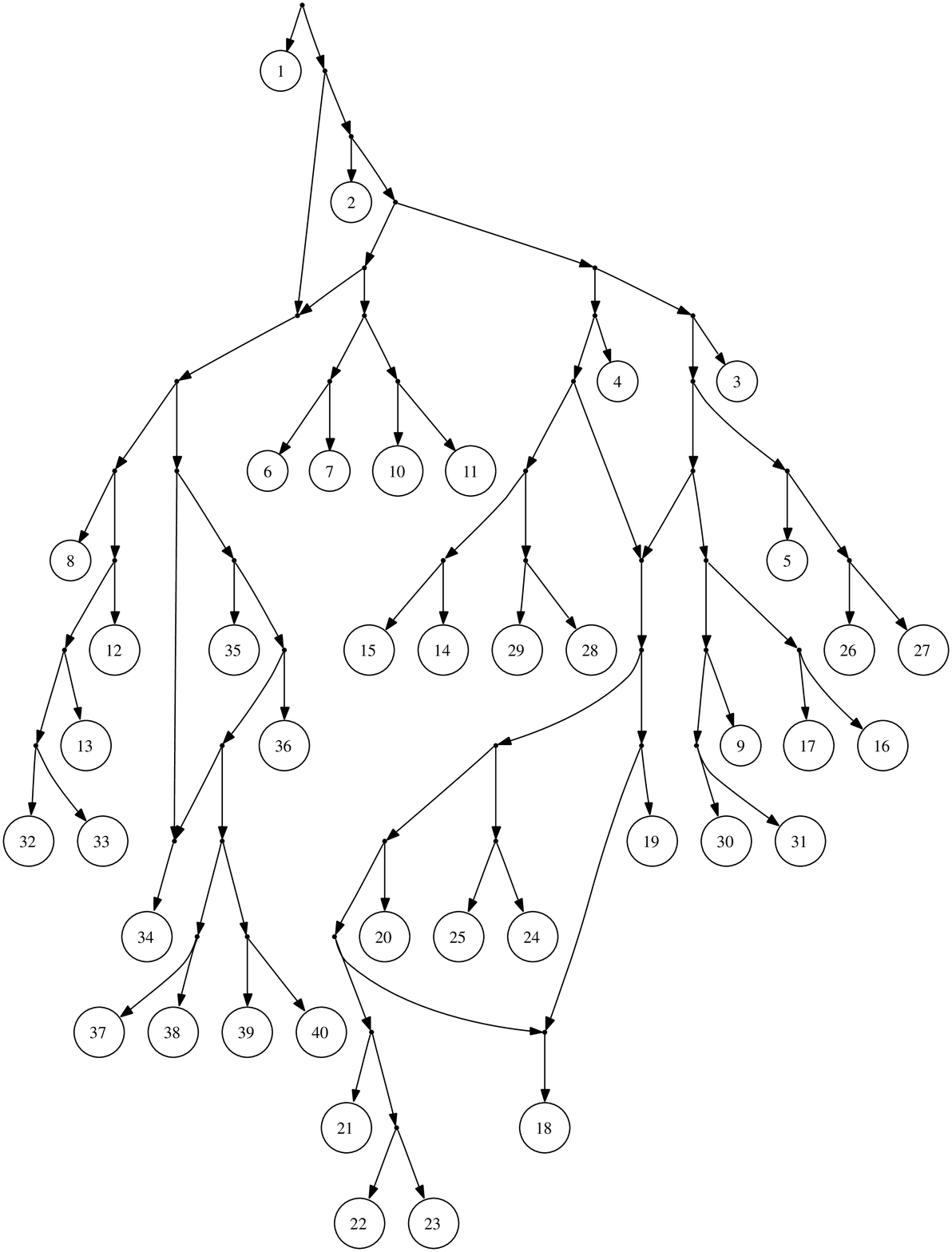}}}\caption{The network constructed by MARLON for the
simulated triplet set $T^*$.}\label{fig:marlon40}
\end{figure}

\noindent From these simulated sequences we computed a set of triplets as follows. We assume that for one sequence it
is known that it is only distantly related to the others. This is called the \emph{outgroup} sequence. For each
combination of three sequences, plus the outgroup sequence, we computed a phylogenetic tree using the maximum
likelihood method PHYML
\cite{phyml}. The output trees of PHYML give a dense triplet set, which we used as input to MARLON.\\

\noindent All simulations gave similar results. Here we describe the results for one specific ``real'' level-1
network, displayed in Figure~\ref{fig:input40}. We obtained the simulated triplet set $T^*$ based on this network by
the procedure described above. For this triplet set MARLON constructed the output network in
Figure~\ref{fig:marlon40}. The constructed network is very similar to the input network (which we assumed to be the
``real'' network). Both networks have four reticulations and also the branching structure is almost identical. The
only differences are all of the following type. The output network contains some subnetworks rooted below a parent of
a reticulation. In some of these cases the input network is a bit different because here the subnetwork is divided
below a path on the cycle, ending in the parent of the reticulation. For example in Figure~\ref{fig:input40} the
leaves 37, 38, 39, 40 are below a path on a cycle consisting of three vertices. However, in the output network in
Figure~\ref{fig:marlon40} these leaves are below a single vertex on the cycle.\\

\noindent Other simulations give similar results. The networks constructed by MARLON are almost identical to the input
networks, except for some small differences that are almost all of the type described above. In one case the output
network also contained an extra reticulation that was not present in the input network. In this case there must have
been triplets in the simulated triplet set that were not consistent with the input network.\\
\\
We conclude that MARLON correctly constructs level-1 networks and works very fast. For simulated data the produced
networks are very close to the ``real'' networks used to generate the simulated sequences. When using real data we
expect the amount of incorrect triplets to be larger and hence the results possibly less impressive. In addition, real
data sets will not always originate from a level-1 network, in which case MARLON will not be able to compute a
solution. This problem will partly be solved in the next section where we show how the approach can be extended to
level-2. However, the main conclusion to be drawn from the experiments is that, if the data is good enough, our method
is indeed able to produce good estimates of evolutionary histories. This for example shows that, when a set of
triplets is computed from sequence data, sufficient information is retained to be able to reconstruct the phylogenetic
network accurately. In addition, MARLON provides a very fast method to combine these triplets into a phylogenetic
network.

\vspace{2cm} \pagebreak

\section{Constructing a Level-2 Network with a Minimum Number of Reticulations}\label{sec:lev2}
This section extends the approach from Section~\ref{sec:lev1} to level-2 networks. We describe how one can find a
level-2 network consistent with a dense input triplet set containing a minimum number of reticulations, or decide that
such a network does not exist.\\
\\
The general structure of the algorithm is the same as in the level-1 case. We loop though all SN-sets $S$ from small
to large and compute an optimal solution $N_S$ for that SN-set, based on previously computed optimal solutions for
included SN-sets. For each SN-set we still consider, like in the level-1 case, the possibility that there are two
cut-arcs leaving the root of $N_S$ and the possibility that this root is in a biconnected component with one
reticulation. However, now we also consider a third possibility, that the root of $N_S$ is in a biconnected component
containing two reticulations.\\
\\
In the construction of biconnected components with two reticulations, we use the notion of
``non-cycle-reachable''-arc, or n.c.r.-arc for short, introduced in \cite{reflections}. We call an arc $a=(u,v)$ an
\emph{n.c.r.-arc} if $v$ is not reachable from any vertex in a cycle. These n.c.r.-arcs will be used to combine
networks without increasing the network level. In addition, we use the notion \emph{highest biconnected
component} to denote the biconnected component containing the root.\\

\begin{figure}[h]\centering
{\centerline{\includegraphics[scale=.8]{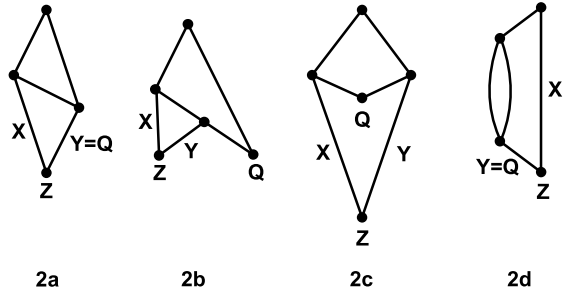}}}\caption{The four possible structures of a biconnected component
containing two reticulations.}\label{fig:level2} \label{fig:2gen}
\end{figure}

\noindent Our complete algorithm is described in detail in Algorithm~\ref{alg:MARLTN}. To get an intuition of why the
algorithm works, consider the four possible structures of a biconnected component containing two reticulations
displayed in Figure~\ref{fig:level2}. Let $X$, $Y$, $Z$ and $Q$ be the sets of leaves indicated in
Figure~\ref{fig:level2} in the graph that displays the form of the highest biconnected component of $N_S$. Observe
that after removing $Z$ in each case $X$, $Y$ and $Q$ become a set of leaves below a cut-arc and hence an SN-set
(w.r.t $T|(S\setminus Z)$). In cases 2a, 2b and 2c the highest biconnected component becomes a cycle, $Q$ the set of
leaves below the highest reticulation and $X$ and $Y$ sets of leaves below highest cut-arcs. We will first
describe the approach for these cases and show later how a similar technique is possible for case 2d.\\
\\
Our algorithm loops through all SN-sets that are a subset of $S$ and will hence at some iteration consider the SN-set
$Z$. The algorithm removes the set $Z$ and computes the SN-sets w.r.t. $T|(S\setminus Z)$. The sets of leaves below
highest cut-arcs (in some optimal solution, if one exists) are now equal to $X,Y,Q$ and the SN-sets that are maximal
under the restriction that they do not contain $X$, $Y$ or $Q$ (by the same arguments as in the proof of
Lemma~\ref{lem:snsets}). Therefore, the algorithm tries each possible SN-set for $X$, $Y$ and $Q$ and in one of these
iterations it will correctly determine the sets of leaves below highest cut-arcs. Then the algorithm computes the
induced set of triplets, where each set of leaves below a highest cut-arc is replaced by a single meta-leaf. All
simple level-1 networks consistent with this induced set of triplets are obtained by the algorithm in \cite{JS2}. Our
algorithm loops through all these networks and does the following for each simple level-1 network $N_1$. Each
meta-leaf $V$, not equal to $X$ or $Y$, is replaced by an optimal network $N_V$, which has been computed in a previous
iteration. To include leaves in $Z$, $X$ and $Y$, we compute an optimal network $N_2$ consistent with $T|(X\cup Z)$ and an
optimal network $N_3$ consistent with $T|(Y\cup Z)$ where in both networks $Z$ is the set of leaves below an n.c.r.-arc. Then
we combine these three networks into a single network like in Figure~\ref{fig:level2example}. A new reticulation is
created and $Z$ becomes the set of leaves below this reticulation. Finally, we check for each constructed network
whether it is consistent with $T|S$. The network with the minimum number of reticulations over all
such networks is the optimal solution $N_S$ for this SN-set.\\
\\
Now consider case 2d. Suppose we remove $Z$ and replace $X$, $Y$ (=$Q$) and each SN-set w.r.t. $T|(S\setminus Z)$ that
is maximal under the restriction that it does not contain $X$ or $Y$ by a single leaf. Then the resulting network
consists of a path ending in a simple level-1 network, with $X$ a child of the root and $Q$ the child of the
reticulation; and each vertex of the path has a leaf as child. Such a network can easily be constructed and
subsequently one can use the same approach as in cases 2a, 2b and 2c. See Figure~\ref{fig:level2example2} for an
example of the construction in case 2d.

\begin{figure}[t]\centering
{\centerline{\includegraphics[scale=.8]{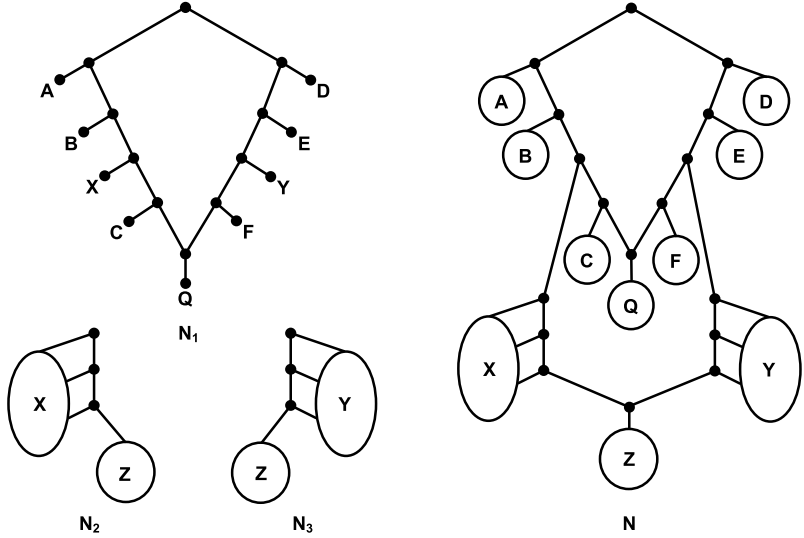}}}\caption{Example of the construction of network $N$ from
$N_1$, $N_2$ and $N_3$.}\label{fig:level2example}
\end{figure}

\begin{figure}[t]\centering
{\centerline{\includegraphics[scale=.8]{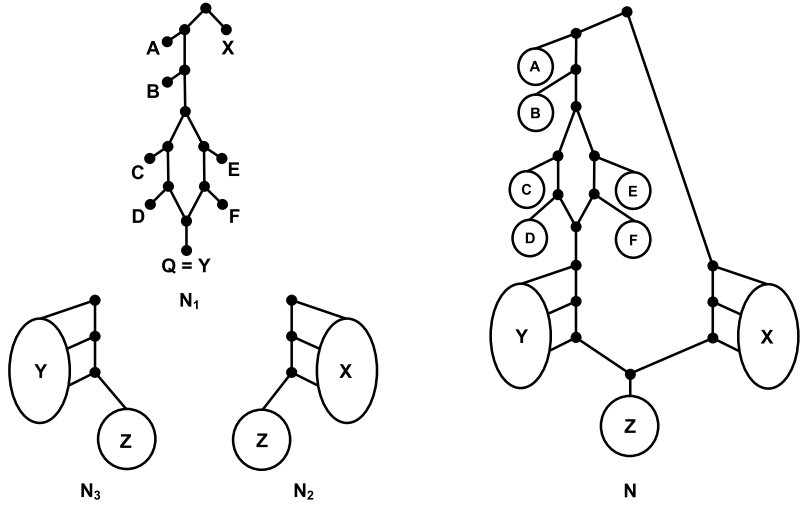}}}\caption{Example of the construction of network $N$ from $N_1$,
$N_2$ and $N_3$ in case 2d.}\label{fig:level2example2}
\end{figure}

\begin{algorithm}[t]
\caption{MARLTN (Minimum Amount of Reticulation Level Two Network)} \label{alg:MARLTN}
\begin{algorithmic} [1]
\STATE - compute the set $SN$ of SN-sets w.r.t. $T$\\
\FOR{$i=1\ldots n$} \FOR{each $S$ in $SN$ of cardinality $i$}

\FOR{each $S'\in SN$ with $S'\subset S$}

\STATE - let $\mathcal{C}$ contain $S'$ and all SN-sets that are maximal under the restriction that they are a subset
of $S$ and do not contain $S'$

\IF{$T\nabla \mathcal{C}$ is consistent with a simple level-1 network $N_1$}

\STATE - construct $N^*$ from $N_1$ by replacing each leaf $V$ by an optimal network $N_V$ constructed in a previous iteration\\

\STATE - $g(S,S')$ is the number of reticulations in $N^*$\\

\ENDIF

\ENDFOR

\IF{there are exactly two SN-sets $S_1,S_2\in SN$ that are maximal under the restriction that they are a strict subset
of $S$}

\STATE - $N^*$ consists of a root connected to the roots of optimal networks $N_{S_1}$ and $N_{S_2}$ that have been
constructed in previous iterations\\

\STATE - $g(S,\emptyset)$ is the number of reticulations in $N^*$\\

\ENDIF


\FOR{each $Z\in SN$ with $Z\subset S$}

\STATE - $T':=T|(S\setminus Z)$\\

\STATE - compute the set $SN'$ of SN-sets w.r.t. $T'$\\

\FOR{each $X,Y,Q\in SN'$}

\STATE - $\mathcal{C}$ is the collection consisting of $X,Y,Q$ and all SN-sets in $SN'$ that are maximal under the
restriction
that they do not include $X,Y$ or $Q$\\

\STATE - construct an optimal network $N_2$ consistent with $T|(X\cup Z)$ such that $Z$ is the set of leaves below an
n.c.r.-arc
$(u,v)$\\

\STATE - construct an optimal network $N_3$ consistent with $T|(Y\cup Z)$ such that $Z$ is the set of leaves below an
n.c.r.-arc $(u',v')$\\

\STATE - construct all simple level-1 networks consistent with $T'\nabla \mathcal{C}$\\

\STATE - construct all networks consistent with $T'\nabla \mathcal{C}$ that consist of a path ending in a simple
level-1 network, with $X$ a child of the root, $Q$ the child of the reticulation; and with a leaf below each
internal vertex of the path\\

\FOR{each network $N_1$ from the networks constructed in the above two lines}

\STATE - construct $N^*$ from $N_1$ by doing the following: replace $X$ by $N_2$, $Y$ by $N_3$ and each other leaf $V$
by an optimal network $N_V$ constructed in a previous iteration, then subdivide $(u,v)$ into $(u,w)$ and $(w,v)$,
delete everything below $u'$ and add an arc $(u',w)$\\

\IF{$N^*$ is consistent with $T|S$}

\STATE - $h(S,X,Y,Z,Q)$ is the number of reticulations in $N^*$\\

\ENDIF

\ENDFOR

\ENDFOR

\ENDFOR


\STATE - $f(S)$ is the minimum of all computed values of $g(S,S')$ and $h(S,X,Y,Z,Q)$\\

\STATE - store network $N_S$, which is a network $N^*$ attaining the minimum number $f(S)$ of reticulations\\

\ENDFOR \ENDFOR

\end{algorithmic}
\end{algorithm}

\begin{theorem}
Given a dense set of triplets $T$, Algorithm MARLTN constructs a level-2 network that is consistent with $T$ (if such
a network exists) and has a minimum number of reticulations in $O(n^9)$ time.
\end{theorem}
\begin{proof}
Consider some SN-set $S$ and assume that there exists an optimal solution $N_S$ consistent with $T|S$. The proof is by
induction on the size of $S$. If the highest biconnected component of $N_S$ contains one reticulation then the
algorithm constructs an optimal solution by the proof of Theorem~\ref{thm:lev1}. Hence we assume from now on that the
highest biconnected component of $N_S$ contains two reticulations.\\
\\
Consider the four graphs in Figure~\ref{fig:level2}. Any biconnected component containing two reticulations is a
subdivision of one of these graphs \cite[Lemma~13]{arxiv}. These graphs are called simple level-2 \emph{generators} in
\cite{arxiv} and $X$, $Y$, $Z$ and $Q$ each label, in each generator, a \emph{side} of the generator, i.e. either an
arc or a vertex with indegree 2 and outdegree 0. Suppose that the highest biconnected component of $N_S$ is a
subdivision of generator $G$ and let the set $X$ ($Y$, $Z$, $Q$ respectively) be defined as the set of leaves
reachable in $N_S$ from a vertex with a parent on the path corresponding to the side labelled $X$ ($Y$, $Z$, $Q$
respectively) in $G$.\\
\\
To find an optimal network consistent with $T|(X\cup Z)$ (or $T|(Y\cup Z)$) such that $Z$ is below an n.c.r.-arc we
can use the following approach. If there are more than two maximal SN-sets then it is not possible. Otherwise, we
create a root and connect it to two networks for the two maximal SN-sets. If one of these maximal SN-sets contains $Z$
as a strict subset then we create a network for this set recursively. For other maximal SN-sets we use the optimal
networks computed in earlier iterations.\\
\\
Given a network $N'$ and a set of leaves $L'$ below a cut-arc $(u,v)$ we denote by $N'\setminus L'$ the network
obtained by removing $v$ and all vertices reachable from $v$ from $N'$, deleting all vertices with outdegree zero and
suppressing all vertices with indegree and outdegree both one.

\begin{claim}[2] There exists an optimal solution $N'$ such that the sets of leaves below highest cut-arcs of $N'\setminus Z$
are $X$, $Y$, $Q$ and the SN-sets w.r.t. $T|(S\setminus Z)$ that are maximal under the restriction that they do not
contain $X$, $Y$ or $Q$.
\end{claim}
\begin{proof}
The highest biconnected component of $N'\setminus Z$ contains just one reticulation and the same arguments can be used
as in the proof of Lemma~\ref{lem:snsets}. \qed
\end{proof}

\noindent Let $N'$ be a network with the property described in the claim above and $\mathcal{C}$ the collection of
sets of leaves below highest cut-arcs of $N'\setminus Z$. At some iteration the algorithm will consider this set
$\mathcal{C}$. Let $T'$ equal $T|(S\setminus Z)$. If we replace in $N'\setminus Z$ each set of leaves below a highest
cut-arc by a single leaf, then we obtain a network consistent with $T'\nabla \mathcal{C}$ which is either a simple
level-1 network or a path ending in a simple level-1 network, with $X$ a child of the root, $Q$ the child of the
reticulation; and each vertex of the path has a leaf as child. The algorithm considers all networks of these types, so
in some iteration it will consider the right one. Let $N^*$ be the network constructed by the algorithm in this
iteration. It remains to prove that $N^*$ (i) is consistent with $T|S$, (ii) contains a minimum number of
reticulations and (iii) is a level-2 network.\\
\\
To prove that $N^*$ is consistent with $T|S$, consider any triplet $xy|z\in T|S$. First suppose that $x,y$ and $z$ are
all in $Z$ or all in the same set of
$\mathcal{C}\setminus\{X,Y\}$. Then $x,y$ and $z$ are elements of some SN-set $S'$ with
$|S'|<|S|$. Triplet $xy|z$ is consistent with the subnetwork $N_{S'}$ by the induction hypothesis and hence with
$N^*$.\\
\\
Now suppose that $x,y$ and $z$ are all in $X\cup Z$ (or all in $Y\cup Z$). Consider the construction of the network
consistent with $X\cup Z$ such that $Z$ is below an n.c.r.-arc. First suppose that at some level of the recursion
there are two maximal SN-sets, each containing leaves from $\{x,y,z\}$. Then it follows that $x$ and $y$ are in one
maximal SN-set and $z$ in the other one, by the definition of SN-set, and hence that $xy|z$ is consistent with the
constructed network. Otherwise, $x,y$ and $z$ are all in some subnetwork $N_S'$ with $|S'|<|S|$ and is $xy|z$
consistent with this subnetwork (by the induction hypothesis) and hence with $N^*$.\\
\\
Now consider any other triplet $xy|z\in T|S$, which thus contains leaves that are below at least two different highest
cut-arcs. Observe that the highest biconnected components of $N^*$ and $N'$ are identical; the only differences
between these networks occur in the subnetworks below highest cut-arcs. Therefore $xy|z$ is consistent with $N^*$
since it is consistent with $N'$.\\
\\
To show that $N^*$ contains a minimum number of reticulations consider any set $S'$ of leaves below a highest cut-arc
$a=(u,v)$ of $N^*$. The subnetwork $N_{S'}$ rooted at $v$ contains a minimum number of reticulations by the induction
hypothesis. Hence $N^*$ contains at most as many reticulations as $N'$, which is an optimal solution.\\
\\
In the networks $N_2$ and $N_3$ is $Z$ the set of leaves below an n.c.r.-arc. This implies that none of the potential
reticulations in these networks end up in the highest biconnected component of $N'$. Therefore, this biconnected
component contains exactly two reticulations. All other biconnected components of $N'$ also contain at most two
reticulations by the induction hypothesis. We thus conclude that $N'$ is a level-2 network.\\
\\
To conclude the proof we analyse the running time of the algorithm. The number of SN-sets is $O(n)$ and hence there
are $O(n)$ choices for each of $S,X,Y,Z$ and $Q$. For each combination of $S,X,Y,Z$ and $Q$ there will be $O(n)$
networks $N^*$ constructed and for each of them it takes $O(n^3)$ time to check if it is consistent with $T|S$ (in
line~23). Hence the overall time complexity is $O(n^9)$. \qed\end{proof}


\section{Constructing Networks Consistent with Precisely the Input Triplet Set}
\label{sec:extreme} In this section we consider the problem MIN-REFLECT-$k$. Given a triplet set $T$, this problem
asks for a level-$k$ network $N$ that is consistent with precisely those triplets in $T$ (if such a network exists)
and has minimum level over all such networks. We will show that this problem is polynomial-time solvable for each
fixed $k$.\\
\\
Recall that we use $T(N)$ to denote the set of all triplets consistent with a network $N$. Furthermore, we say that a
network $N$ \emph{reflects} a triplet set $T$ if $T(N)=T$. The problem MIN-REFLECT-$k$ thus asks for a minimum level
network $N$ that reflects an input triplet set $T$, for some fixed upper bound $k$ on the level of $N$. Note that, if
$N$ reflects $T$, that $N$ is in general not uniquely defined by $T$. There are, for example, several distinct simple
level-2 networks that reflect the triplet set $\{ xy|z, xz|y, zy|x \}.$\\
\\
Note that MIN-REFLECT-0 can be easily solved by using the algorithm of Aho et al. \cite{aho}. This follows because if
a tree $N$ is consistent with a dense set of triplets $T$, then $N$ is unique \cite{JS1} (and $T(N) = T$).\\
\\
\textbf{Theorem \ref{thm:all}. }\emph{Given a dense set of triplets $T$, it is possible to construct all
simple level-$k$ networks consistent with $T$ in time $O(|T|^{k+1})$.}\\
\\
\textbf{Lemma \ref{singleton}. }\emph{Let $N$ be any simple network. Then all the nontrivial SN-sets of $T(N)$ are singletons.}\\
\\
\noindent As we will show, the above lemma allows us to solve the problem MIN-REFLECT-$k$ by recursively constructing simple
level-$k$ networks, which we can do by Theorem~\ref{thm:all}. This leads to the algorithm MINPITS-$k$ (MINimum level
network consistent with Precisely the Input Triplet Set).\\
\\
\textbf{Theorem \ref{thm:minpits}. }\emph{Given a set of triplets $T$, Algorithm MINPITS-$k$ solves MIN-REFLECT-$k$ in time $O(|T|^{k+1})$, for any fixed
$k$.}

\subsection{Constructing all simple level-$k$ networks in polynomial time}

To prove Theorem \ref{thm:all} we first need several utility lemmas. Recall the concept of a \emph{simple level-k
generator} \cite[Section~3.1]{arxiv}. (See also the proof of Theorem~\ref{alg:MARLTN}.)

\begin{lemma}
\label{lem:genbound}
Let $G$ be a simple level-$k$ generator. Then $G$ contains $O(k)$ vertices and $O(k)$ arcs.
\end{lemma}
\begin{proof}
Suppose $G$ contains $s$ split vertices, $k$ reticulation vertices and 1 root. Then the total
indegree equals $s+2k$ and the total outdegree is at least $2s+2$. Given that the
total indegree equals the total outdegree we get that $s+2k \geq 2s+2$ and hence that
$s \leq 2k-2$. So the total number of vertices is at most $3k-1$. All the vertices have
at most degree 3 so there are at most $(9k-3)/2$ arcs. \qed
\end{proof}

\begin{lemma}
\label{lem:edgebound}
Let $N$ be a level-$k$ network on $n$ leaves, where $k$ is fixed. Then $N$ contains $O(n)$ vertices
and $O(n)$ arcs.
\end{lemma}
\begin{proof}
By the definition of a phylogenetic network we can view $N$ as a rooted, directed ``component tree'' $B(N)$ of biconnected components where every
internal vertex of $B(N)$ represents a simple level-${\leq}k$ subnetwork of $N$ (or a single vertex),
and incident arcs of an internal vertex represent arcs incident to the corresponding
biconnected component. $B(N)$ has $n$ leaves, and every internal vertex
has at least two outgoing arcs. $B(N)$ is a tree so it has at
most $n-1$ internal vertices
and thus at most $2n-1$ vertices in total, and at most $2n-2$ arcs. Each internal vertex represents
a simple level-${\leq}k$ generator with at most $(9k-3)/2$ arcs. Every outgoing arc
raises the number of arcs (and vertices) inside the component by at most 1. So the total number
of arcs in $N$ is bound above by $(2n-2) + (n-1)(9k-3)/2 + (2n-2)$, which is $O(n)$, and
the total number of vertices by $n + (n-1)(3k-1) + (2n-2)$, also $O(n)$. \qed
\end{proof}

\noindent Let $N$ be a network with at least one reticulation vertex, and let $v$ be the child of a reticulation
vertex in $N$. If $v$ has no reticulation vertex as a descendant, then we call the subnetwork rooted at $v$ a
\emph{Tree hanging Below a Reticulation vertex (TBR)}. We additionally introduce the notion of the \emph{empty} TBR,
which corresponds to the situation when a reticulation vertex has no outgoing arcs. This cannot happen in a normal
network but as explained shortly it will prove a useful abstraction.

\begin{observation}
\label{obs:tbr}
Every network $N$ containing a reticulation vertex contains at least one TBR.
\end{observation}
\begin{proof}
Suppose this is not true. Let $v$ be the child of a reticulation vertex in $N$ with maximum distance from the root.
There must exist some vertex $v' \neq v$ which is a child of a reticulation vertex and which is a descendent of $v$.
But then the distance from the root to $v'$ is greater than to $v$, contradiction. $\Box$
\end{proof}

\noindent Note that, because a TBR is (as a consequence of its definition) below a cut-arc, there exists an SN-set $S$
w.r.t. $T$ such that $T|S$ is consistent with (only) the TBR. An SN-set $S$ such that $T|S$ is consistent with a tree,
we call a \emph{CandidateTBR} SN-Set. Every TBR of $N$ corresponds to some CandidateTBR SN-Set of $T$, but the
opposite is not necessarily true. For example, a singleton SN-set is a CandidateTBR SN-Set, but it might not be the
child of a reticulation vertex in $N$.\\
\\
We abuse definitions slightly by defining the \empty{empty} CandidateTBR SN-Set, which will correspond to the empty
TBR. (This is abusive because the empty set is not an SN-set.) Furthermore we define that every triplet set $T$ has an
empty CandidateTBR SN-Set.

\begin{observation}
\label{obs:gettbr}
Let $T$ be a dense set of triplets on $n$ leaves. There are at most $O(n)$ CandidateTBR SN-sets. All
such sets, and the tree that each such set represents, can be found in total time $O(n^3)$.
\end{observation}
\begin{proof}
First we construct the SN-tree for $T$, this takes time $O(n^3)$. There is a bijection between the SN-sets of $T$ and the vertices of
the SN-tree. (In the SN-tree, the children of an SN-set $S$ are the maximal SN-sets of $T|S$.)
Observe that a vertex of the SN-tree is a CandidateTBR SN-set if and only if it is a singleton SN-set \emph{or}
it has in total two children, and both are CandidateTBR SN-sets. We can thus
use depth first search to construct all the CandidateTBR SN-sets; note that this is also sufficient
to obtain the trees that the CandidateTBR SN-sets represent, because (for trees) the structure of
the tree is identical to the nesting structure of its SN-sets. Given that there are only $O(n)$ SN-sets,
the running time is dominated by construction of the SN-tree. $\Box$
\end{proof}

\begin{theorem}
\label{thm:all}
Given a dense set of triplets $T$, it is possible to construct all simple level-$k$ networks consistent with $T$ in time $O(|T|^{k+1})$.
\end{theorem}
\begin{proof}
We claim that algorithm SL-$k$ does this for us. First we prove correctness. The high-level idea is as follows.
Consider a simple level-$k$ network $N$. From Observation \ref{obs:tbr} we know that $N$ contains at least one TBR.
(Given that $N$ is simple we know that all TBRs are equal to single leaves. That is why the outermost loop of the
algorithm can restrict itself to considering only single-leaf TBRs.) By looping through all CandidateTBR SN-sets we
will eventually find one that corresponds to a real TBR. If we remove this TBR and the reticulation vertex from which
it hangs, and then suppress any resulting vertices with both indegree and outdegree equal to 1, we obtain a new
network (not necessarily simple) with one fewer reticulation vertex than $N$. Note that this new network might not be
a ``real'' network in the sense that it might have reticulation vertices with no outgoing arcs. Repeating this $k$
times in total we eventually reach a tree which we can construct using the algorithm of Aho et al. (and is unique, as
shown in \cite{JS1}). From this tree we can reconstruct the network $N$ by reintroducing the TBRs back into the
network (each TBR below a reticulation vertex) in the reverse order to which we found them. We don't, however, know
exactly where the reticulation vertices were in $N$, so every time we reintroduce a TBR back into the network we
exhaustively try every pair of arcs (as the arcs which will be subdivided to hang the reticulation vertex, and thus
the TBR, from.) Because we try every possible way of removing TBRs from the network $N$, and every possible way of
adding them back, we will eventually reconstruct $N$.\\
\\
The role of the dummy leaves in SL-$k$ is linked to the empty TBRs (and their corresponding empty CandidateTBR
SN-Sets.) When a TBR is removed, it can happen (as mentioned above) that a network is created containing reticulation
vertices with no outgoing arcs. (For example: when one of the parents of a reticulation vertex from which the TBR
hangs, is also a reticulation vertex.) Conceptually we say that there \emph{is} a TBR hanging below such a
reticulation vertex, but that it is empty. Hence the need in the algorithm to also consider removing the empty TBR. If
this happens, we will also encounter the phenomenon in the second phase of the algorithm, when we are re-introducing
TBRs into the network. What do we insert into the network when we reintroduce an empty TBR? We use a dummy leaf as a
place-holder, ensuring that every reticulation vertex always has an outgoing arc. The dummy leaves can be removed once
that outgoing arc is subdivided later in the algorithm, or at the end of the algorithm, whichever happens sooner. The
dummy root, finally, is needed for when there are no leaves on a side (see  \cite[Section~3.1]{arxiv}) connected to
the root.\\
\\
We now analyse the running time.  For $k \in \{1,2\}$ we can actually generate all simple level-1 networks in time
$O(n^3)$ using the algorithm in \cite{JS2}, and all simple level-2 networks in time $O(n^8)$ using the
algorithm in \cite{arxiv}. For $k \geq 3$ we use SL-$k$. From Observation \ref{obs:gettbr} we know that each execution of
$FindCandidateTBRs$ (which computes all TBRs in a dense triplet set plus the empty TBR) takes $O(n^3)$ time and
returns
at most $O(n)$ TBRs. Operations such as computing $T'_i$, and the construction of the tree $N'_{k}$, all require time bounded above by
$O(n^3)$. The for loops when we ``guess'' the TBRs are nested to a depth of $k$. The for loops
when we ``guess'' pairs of arcs from which to hang TBRs, are also nested to a depth of $k$. (There will only
be $O(n)$ arcs to choose from.) Checking whether
$N'$ is consistent with $T$, which we do inside the innermost loop of the entire algorithm, takes time $O(n^3)$ \cite[Lemma~2]{byrka}.
So the running time is $O( n(n^3 + n(n^3 \ldots n(n^3 + n^{2k+3} ))$ which is $O(n^{3k+3})$. \qed
\end{proof}

\begin{algorithm}[H]
\caption{SL-$k$ (Construct all Simple Level-$k$ networks)} \label{alg:SLK}
\begin{algorithmic} [1]
\STATE $Net:=\emptyset$\\
\STATE $TBR_1 := L(T)$\\
\FOR{each leaf $b_1 \in TBR_1$}
\STATE $L'_1 := L(T) \setminus \{ b_1 \} $\\
\STATE $T'_1 := T | L'_1$\\
\STATE $TBR_2 := FindCandidateTBRs(T'_1)$\\
\FOR{each $b_2 \in TBR_2$}
\STATE ...\\
\COMMENT{ Continue nesting \textbf{for} loops to a depth of $k$. }\\
\STATE ...\\
\STATE $TBR_k := FindCandidateTBRs(T'_{k-1})$\\
\FOR{each $b_k \in TBR_k$}
\STATE $L'_k := L'_{k-1} \setminus b_k$\\
\STATE $T'_k := T'_{k-1} | L'_k$\\

\COMMENT{ At this point we have finished ``guessing'' where the TBRs are, }\\
\COMMENT{ and $(\{b_1\}, b_2, ..., b_k)$ is a vector of (possibly empty) subsets of $L(T)$. }\\
\COMMENT{ We now ``guess'' all possible ways of hanging the TBRs back in. }\\

\IF{$L'_k$ contains 2 or more leaves} \STATE build the unique tree $N'_{k} = (V,A)$ consistent with $T'_k$ if it
exists (see \cite{JS1}) \ELSE
\STATE If $L'_k$ contains 1 leaf $\{x\}$, let $N'_{k}$ be the network comprising the single leaf $\{x\}$\\
\STATE If $L'_k$ contains 0 leaves, let $N'_{k}$ be the network comprising a single, new dummy leaf\\
\ENDIF
\STATE $V := V \cup \{r'\}; A := A \cup \{(r',r ) \}$ \COMMENT{ with $r$ the root of $N'_{k}$ and $r'$ a new dummy root }\\

\STATE Let $H(b_k)$ be the unique tree consistent with $b_k$\\
\COMMENT{ Note: $H(b_k)$ is a single vertex if $|b_k|=1$ and empty if $|b_k| = 0$. }

\FOR{every two arcs $a^{1}_k$, $a^{2}_k$ in $N'_{k}$ (not necessarily distinct)}
\STATE Let $p$ (respectively $q$) be a new vertex obtained by subdividing $a^{1}_k$ (respectively $a^{2}_k$)\\
\STATE Connect $p$ and $q$ to a new reticulation vertex $ret_{k}$\\
\STATE Hang $H(b_k)$ (or a new dummy leaf if $H(b_k)$ is empty) from $ret_{k}$\\
\IF{$a^{1}_k$ (or $a^{2}_k$) was the arc above a dummy leaf $d$}
\STATE{Remove $d$ and if its former parent has indegree and outdegree 1, suppress that}\\
\ENDIF \STATE Let $N'_{k-1}$ be the resulting network
\STATE Let $H(b_{k-1})$ be the unique tree consistent with $b_{k-1}$\\
\FOR{every two arcs $a^{1}_{k-1}$, $a^{2}_{k-1}$ in $N'_{k-1}$ (not necessarily distinct)}
\STATE ...\\
\COMMENT{ Continue nesting \textbf{for} loops to a depth of $k$. }\\
\STATE ...\\
\STATE Let $N'_{1}$ be the resulting network\\
\STATE Let $H(b_1)$ be the tree consisting of only the single vertex $b_1$\\
\FOR{every two arcs $a^{1}_1$, $a^{2}_1$ in $N'_{1}$ (not necessarily distinct)}
\STATE Let $p$ (respectively $q$) be a new vertex obtained by subdividing $a^{1}_1$ (respectively $a^{2}_1$)\\
\STATE Connect $p$ and $q$ to a new reticulation vertex $ret_{1}$\\
\STATE Hang $H(b_1)$ from $ret_{1}$\\
\IF{$a^{1}_1$ (or $a^{2}_1$) was the arc above a dummy leaf $d$}
\STATE{Remove $d$ and if its former parent has indegree and outdegree 1, suppress that}\\
\ENDIF
\COMMENT{ This is the innermost loop of the algorithm. }\\
\STATE{Let $N'$ be the resulting network}\\
\STATE{Remove the dummy root $r'$ from $N'$}\\
\STATE{Remove (and if needed suppress former parents of) any remaining dummy leaves in $N'$}\\
\IF{ $N'$ is a simple level-$k$ network consistent with $T$ }
\STATE{ $Net := Net \cup \{N'\}$ }\\
\ENDIF \ENDFOR \ENDFOR \ENDFOR \ENDFOR \ENDFOR \ENDFOR

\STATE \textbf{return} $Net$\\
\end{algorithmic}
\end{algorithm}

\clearpage

\begin{corollary}
\label{cor:allextreme}
For fixed $k$ and a triplet set $T$ it is possible to generate in time $O(n^{3k+3})$ all
simple level-$k$ networks $N$ that reflect $T$.
\end{corollary}
\begin{proof}
The algorithm SL-$k$ (or, for that matter, the algorithms from \cite{JS2}\cite{arxiv}) can easily be adapted for this purpose: we change
in line 43 ``network consistent with $T$'' to ``network that reflects $T$''. The running time
is unchanged because, whether we are checking consistency or reflection, the implementation of \cite[Lemma~2]{byrka}
implicitly generates $T(N')$. \qed
\end{proof}

\subsection{From simple networks that reflect, to general networks that reflect}

For a triplet $t = xy|z$ and a network $N$, we define an \emph{embedding} of $t$ in
$N$ as any set of four paths $(q \rightarrow x, q \rightarrow y, p \rightarrow q, p \rightarrow z)$ which, except for 
their endpoints, are mutually vertex disjoint, and where $p \neq q$. We say that the vertex $p$ is the \emph{summit} of the 
embedding. Clearly, $t$ is consistent with $N$ if and only if there is at least one embedding of $t$ in $N$.

\begin{lemma}
\label{singleton}
Let $N$ be any simple network. Then all the nontrivial SN-sets of $T(N)$ are singletons.
\end{lemma}
\begin{proof}
We prove the lemma by contradiction. Assume thus that there is some SN-set $S$ of $T(N)$ such that $1 < |S| < |L(N)|$.

Let $r$ be the root of $N$. An \emph{in-out root embedding} is an embedding of any triplet $xz|y$ with
$x,z\in S$ and $y\notin S$ that has $r$ as its summit. We begin by proving that an in-out root embedding exists in $N$.
Suppose by contradiction this is not true. The triplet $xz|y$ must be in $T(N)$
because $T(N)$ is dense and $y \not \in S$. Consider then a triplet embedding $(q \rightarrow x, q \rightarrow z, p
\rightarrow q, p \rightarrow y)$ where $p \neq r$, $x, s \in S$ and $y \not \in S$. We assume without loss of
generality that this embedding minimises (amongst all such embeddings) the length of the shortest directed path from
$r$ to $p$. Let $P$ be any shortest directed path from $r$ to $p$. Now, suppose some directed path $Q$ begins
somewhere on the path $P$, and intersects with the path $p \rightarrow y$. This is not possible because it would
contradict the minimality of the length of $P$. For the same reason, $Q$ may not intersect with the path $p
\rightarrow q$. $Q$ may also not intersect with (without loss of generality) the path $q \rightarrow x$ because this
would mean $y \in S$. We conclude that such a path $Q$ either terminates at a leaf $l$, or re-intersects with the path
$P$. It cannot terminate at a leaf $l$ because then we either violate the minimality of the length of $P$ (because we
obtain an embedding of $xz|l$ that has a summit closer to the root than $p$), or we have that $y \in S$. We conclude
that all outgoing paths from $P$ must re-intersect with $P$. However, given that $P$ includes the root, and that in a
directed acyclic graph every vertex is reachable by a directed path from the root, it follows that the last arc on the
path $P$ must be a cut-arc. But this violates the biconnectivity of $N$, contradiction. We conclude that there exists at
least one in-out root embedding in $N$.

Let $(q \rightarrow x, q \rightarrow z, r
\rightarrow q, r \rightarrow y)$ be any in-out root embedding. We observe that the path $r \rightarrow y$ must contain
at least one internal vertex, by biconnectivity. Also, at least one of $q \rightarrow x$ and $q\rightarrow y$ must
contain an internal vertex, because it is not possible for a vertex in a simple network to have two leaf children.

We now argue that there must exist a \emph{twist cover} of the path $r \rightarrow q$. This is defined as a non-empty
set $C$ of undirected paths (undirected in the sense that not all arcs need to have the same orientation) where (i) all
paths in $C$ are arc-disjoint from the in-out root embedding, (ii) exactly
one path starts at an internal vertex $s$ of (without loss of generality) $q \rightarrow z$, (iii)
exactly one path ends at an internal vertex $t$ of $r \rightarrow y$, (iv) all other start and endpoints
of the paths in $C$ lie on $r \rightarrow q$ and (v) for every vertex $v$ of the path $r \rightarrow q$ (including
$r$ and $q$), there is at least one path in $C$ that has its startpoint to the left of $v$, and its endpoint to the right.
Property (v) is crucial because it says (informally) that every vertex on $r \rightarrow q$ is ``covered'' by some
path that begins and ends on either side of it and is arc-disjoint from the embedding. The length of $C$ is defined to be the 
sum of the number of 
arcs in each path in $C$.

Suppose however that a twist cover does not exist. We define a \emph{partial} twist cover as one that satisfies all properties 
of a twist
cover except property (v). Partial twist covers have thus at least one vertex on $r \rightarrow q$ that is not covered. (To see 
that
there always exists at least one partial twist cover note that properties (ii) and (iii) in particular are satisfied by the 
fact that neither the 
removal of $q$ or $r$ can be allowed to disconnect $N$.) So let $C$ be
the partial twist cover with the maximum number of uncovered vertices. Let $d$ be the uncovered vertex that is closest to
$q$. If we removed $d$ we would, by definition, disconnect the union of the paths in $C$ with the in-out root embedding into a 
left part $G$ and a right part $H$. The vertex $d$ does not, however, disconnect $N$, so there
must be some path $P$ not in $C$ that begins somewhere in $G$ and ends somewhere in $H$. If $P$ has its startpoint on a path $X 
\in C$
(where $X$ will be in $G$) and/or an endpoint on a path $Y$ (where $Y$ will be in $H$) then these paths can be ``merged''
into a new path that strictly increases the number of vertices covered. The merging occurs as follows. We take the union of
the arcs in $P$ with those in $X$ and/or $Y$ and discard superfluous arcs until we obtain a path that covers a strict superset
of the union of the vertices covered by $X$ and/or $Y$. (In particular, the fact that $P$ begins in $G$ and ends in $H$ means 
that
the vertex $d$ becomes covered.)
In this way we obtain a new partial twist-cover with fewer uncovered vertices, contradiction. If $P$ has both its startpoint 
and
endpoint on vertices of $r \rightarrow q$ that are not on paths in $C$, then $P$ can be added to the set $C$ and this extends 
the number of covered vertices,
contradiction. If $P$ begins and/or ends elsewhere on the embedding then $P$ can be added to $C$
which again increases the number of vertices covered, contradiction. (If $P$ begins on, without loss of generality, $q 
\rightarrow z$ then
it becomes the new property-(ii) path and the old property-(ii) path should be discarded. Symmetrically, if $P$ ends on $r 
\rightarrow y$ then
it becomes the new property-(iii) path and the old property-(iii) path should be discarded.)

We conclude that for every in-out root embedding
there thus exists a twist cover, and in particular a minimum length twist cover.

We observe that a minimum-length twist cover $C$ has a highly regular, interleaved structure. This regularity follows
because it cannot contain paths that completely contain other paths (simply discard the inner path) and
if two paths $X, Y \in C$ have startpoints that are both covered by some other path $Z \in C$, and (without loss of generality) 
$X$ reaches further right
than $Y$, then we can simply discard $Y$. For similar reasons minimum-length twist covers are vertex- and arc-disjoint. In 
Figure \ref{fig:twists} we show several
simple examples of twist covers exhibiting this regular structure (although it should be noted that minimum-length twist covers 
can contain arbitrarily
many paths.)

Let $C$ be the twist cover of minimum length ranging over all in-out root embeddings of $xz|y$ where $x,z \in S$ and $y \not 
\in S$.
Note that if $C$ contains exactly one path, which is a directed path, then
(irrespective of the path orientation) $y \in S$, contradiction.
The high-level idea is to show that we can always find, by ``walking'' along the paths in $C$, a new in-out root
embedding and twist cover that is shorter than $C$, yielding a contradiction.
Let $X$ be the path in $C$ that begins at $s$, and consider the arc on this path incident to $s$.
The first subcase is if this arc is directed away
from $s$. Continuing along the path we will eventually encounter an
arc with opposite orientation. This must occur before any
intersection of $X$ with the path $r\rightarrow q$ because otherwise
there would be a directed cycle. Let $v$ be the vertex between these
two arcs, $v$ is a reticulation vertex with indegree 2 and outdegree
1, so there must exist a directed path $Q$ leaving $v$ which
eventually reaches a leaf $m$. If $Q$ intersects with $r \rightarrow
y$ then we have that $y \in S$, contradiction. If $Q$ intersects
with either $q \rightarrow x$ or $q \rightarrow z$ then we obtain a
new in-out root embedding of $xz|y$ and a new twist cover for
that embedding that is shorter than $C$, contradiction. If $Q$ does not intersect with the embedding at all,
then it must be true that $m \in S$ (because the triplet $mz|x$ is
in the network). But then we have an in-out root embedding of the
triplet $zm|y$ with twist cover shorter than $C$, contradiction.

The second subcase (see Figure \ref{fig:thelemma}) is when the first arc of $X$ is entering $s$. There must exist some
directed path $R$ from $r$ to $s$ that uses this arc.
The fact that $r$ is the summit of the embedding means that at some point this directed
path departs from the embedding. Let $w$ be the vertex where $R$
departs from the embeddings for the last time. If $w$
is on the path $r\rightarrow y$ then it follows that $y\in S$, a
contradiction. If $w$ is on one of the paths $q\rightarrow x$ or
$q\rightarrow z$ then leads to a new in-out root embedding with
shorter twist cover, a contradiction. The last case is when $w$
lies on the path $r\rightarrow q$. We create a new in-out
root embedding by using the part of $R$ reachable from $w$,
as an alternative route to $z$. In this way $w$ becomes
the ``q'' vertex of the new embedding (denoted $q'$ in the
figure). To see that we also obtain a new twist cover, note principally that paths in $C$ that
covered $w$ become legitimate candidates for property-(ii)
paths in the new twist cover; in the figure $s'$ denotes the beginning
of the property-(ii) path in the new cover. (Such a path can however partly overlap with $R$. In this case
it is necessary to first remove the part that overlaps with $R$.)
We can furthermore discard all paths from
$C$ that covered $w$ except for the one with endpoint furthest
to the right. Even if this means that no paths from $C$ are
discarded (this happens when $w$ is to the left of all
the paths in $C$ that have their beginning points on $r \rightarrow q$) 
we still get a twist cover at least one edge smaller than $C$, because 
(in particular) the first edge of $X$ is no longer needed in $C$.
In any case, contradiction. \qed
\end{proof}

\begin{figure}[h]
\centering
\includegraphics[scale=0.6]{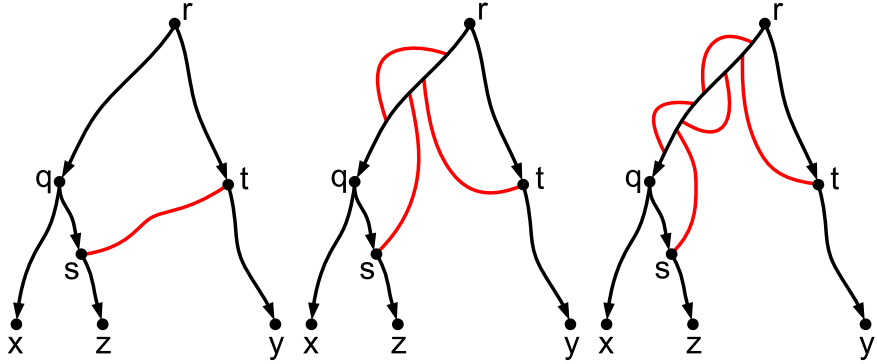}
\caption{Several examples of twist covers (the red, undirected paths) from the proof of Lemma \ref{singleton}. Note that these 
exhibit the
regular, interleaved structure associated with minimum-length twist covers.}
\label{fig:twists}
\end{figure}

\begin{figure}[h]
\centering
\includegraphics[scale=0.6]{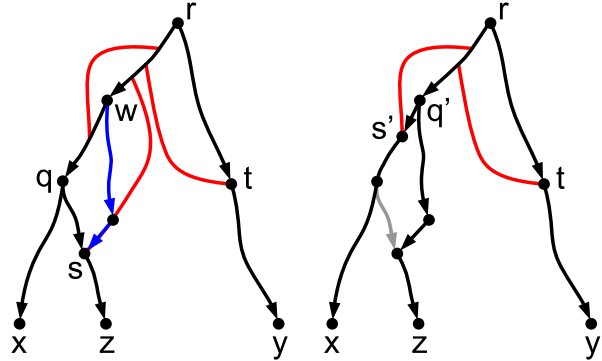}
\caption{The case in the proof of Lemma \ref{singleton} where the arc incident to $s$ is
incoming.}
\label{fig:thelemma}
\end{figure}

\begin{corollary}
Let $T$ be a set of triplets, and suppose there exists a simple network $N$ that reflects $T$.
Let $N'$ be any network that reflects $T$. Then $N'$ is also simple.
\end{corollary}
\begin{proof}
If $N'$ is not simple then it contains a cut-arc below which at least two leaves can be found. We know
\cite[Lemma~3]{arxiv} that every cut-arc of a network consistent with $T$ defines an SN-set w.r.t. $T$ equal to the
set of leaves below it. But all the SN-sets of $T$ are singletons, contradiction. \qed
\end{proof}

\noindent Let $T$ be a reflective set of triplets and $N$ be a network that reflects $T$. Define $Collapse(N)$ as the
network obtained by, for each highest cut-arc $a = (u, v)$, replacing $v$ and everything below it by a single new leaf
$V$, which we identify with the set of leaves below $a$. Let $L'$ be the leaf set of $Collapse(N)$. We define a new
set of triplets $T'$ on the leaf-set $L'$ as follows: $XY|Z \in T'$ if and only if there exists $x \in X, y \in Y$ and
$z \in Z$ such that $xy|z \in T$. We write $T' =  CutInduce(N,T)$ as shorthand for the above.

\begin{observation}
\label{obs:triv}
Let $T$ be a reflective set of triplets, and let $N$ be some network that reflects $T$.
Then (1) $T$ is dense, (2) $T' = CutInduce(N,T)$ is also reflective, and $Collapse(N)$ reflects $T'$
and (3) the maximal SN-sets of $T'$, which are in 1:1 correlation with the maximal SN-sets of $T$,
are all singletons.
\end{observation}
\begin{proof}
The proof of (1) is trivial. For (2), observe firstly that by \cite[Lemma~11.2]{arxiv} (the entire
lemma generalises easily to higher level networks) the network $Collapse(N)$ is simple and consistent with $T'$. Suppose
however there is some triplet $XY|Z \in T(Collapse(N))$ that
is not in $T'$. But we know then that for any $x \in X, y \in Y, z \in Z$ the triplet $xy|z$ is
in $T(N)$, and thus also in $T$, but which would mean that $XY|Z$ is in $T'$. For (3) we argue as follows.
By \cite[Lemma~11.3]{arxiv} it follows that there is a 1:1 correlation between the maximal SN-sets of $T$ and those of $T'$.
Combining the fact that $T'$ is reflective and that $Collapse(N)$ is a simple network that reflects
$T'$, we can conclude from Lemma \ref{singleton} that the maximal SN-sets of $T'$ are all singletons.  \qed
\end{proof}

\begin{lemma}
\label{lem:recurs}
Let $T$ be a reflective set of triplets and let $SN$ be the maximal SN-sets of $T$. Let $T' = T \nabla SN$ be the set of
triplets induced by $SN$. Let $N'$ be a simple network of minimum level that reflects $T'$. Then replacing each leaf $V$ of $N'$
by a network that reflects $T|V$ and is of minimum level amongst such networks, yields a network $N$ that reflects $T$ and is
of minimum level.
\end{lemma}
\begin{proof}
We first prove that $N$ reflects $T$. Let $N^{0}$ be any network that reflects $T$. Every highest cut-arc of $N^{0}$
corresponds exactly to a maximal SN-set w.r.t. $T$. Hence $T' = T \nabla SN = CutInduce(N^{0},T)$. Note also that for
a maximal SN-set $S \in SN$, the set of triplets $T|S$ is reflective: the subnetwork of $N^{0}$ below the cut-arc
corresponding to $S$ reflects $T|S$. (This ensures that the recursive step does find some network.) We know from
Theorem 3 of \cite{arxiv} that the network $N$ is at least consistent with $T$. But suppose there exists some triplet
$xy|z \in N$ that is not in $T$. By construction it cannot be the case that $x,y,z$ are all below the same highest
cut-arc in $N$ (equivalently: in the same maximal SN-set w.r.t. $T$). If exactly two of the leaves are below the same
highest cut-arc, then we see immediately that $xy|z$ is in $T = T(N^{0})$. Suppose then that all three leaves are from
different maximal SN-sets $X$, $Y$, $Z$ respectively. Then $XY|Z \in T'$ because $T' = T(N')$. So we can conclude the
existence of some $x' \in X, y' \in Y, z' \in Z$ such that $x'y'|z' \in T$. But it then also follows that $xy|z \in
T$. We now prove that $N$ is of minimal level. Observe that all networks $N_0$ that reflect $T$ have the same set of
highest cut-arcs, in the sense that each highest cut-arc always corresponds to exactly one maximal SN-set w.r.t. $T$.
Given that the subnetworks created for the $T|V$ are of minimal level, and that the level of $N$ is equal to the
maximum level ranging over $N'$ and all subnetworks below highest cut-arcs, it follows that $N$ is of minimum level if
and only if $N'$ is of minimum level. \qed
\end{proof}

\begin{algorithm}[t]
\caption{MINPITS-$k$ (MINimum level network consistent with Precisely the Input Triplet Set)} \label{alg:L2}
\begin{algorithmic} [1]

\STATE $N:=\emptyset$\\
\STATE compute the set $SN$ of maximal SN-sets of $T$\\
\IF{$|SN|=2$}
\STATE $N$ consists of a root connected to two leaves: the elements of $SN$\\
\ELSE
\IF{there exists a simple level-$\leq k$ network that reflects $T\nabla SN$}
\STATE let $N$ be such a network of minimum level\\
\ENDIF
\ENDIF
\FOR{each leaf $V$ of $N$}
\STATE recursively create a level-$k$ network $N_V$ of minimal level that reflects $T|V$\\
\ENDFOR
\IF{$N\neq\emptyset$ and all $N_V\neq\emptyset$}
\STATE replace each leaf $V$ of $N$ by the recursively created $N_V$.
\STATE \textbf{return} $N$ \\
\ELSE
\STATE \textbf{return} $\emptyset$ \\
\ENDIF

\end{algorithmic}
\end{algorithm}

\begin{theorem}
\label{thm:minpits}
For fixed $k$ we can solve \textsc{MIN-REFLECT-$k$} in time $O(|T|^{k+1}.)$
\end{theorem}
\begin{proof}
From Observation \ref{obs:triv} we know that there is no solution if $T$ is not dense, and this can be checked in
$O(|T|)$ time. We henceforth assume that $T$ is dense. For $k=0$ we can simply use the algorithm of Aho et al., which
(with an advanced implementation \cite{JS1}) can be implemented to run in time $O(n^3)$, which is $O(|T|)$. For $k
\geq 1$ we use algorithm \textsc{MINPITS-$k$}. Correctness of the algorithm follows from Lemma \ref{lem:recurs}. It
remains to analyze the running time. A simple level-$k$ network (that reflects the input) can be found (if it exists)
in time $O(n^{3k+3})$ using algorithm SL-$k$. (To find the simple network of minimum level we execute in order SL-1,
SL-2, ..., SL-$k$ until we find such a network. This adds a multiplicative factor of $k$ to the running time but this
is absorbed by the $O(.)$ notation for fixed $k$.) Therefore, lines 6 and 7 of \textsc{MINPITS-$k$} take
$O(|SN|^{3k+3})$ time. At every level of the recursion the computation of the maximal SN-sets of $T$ can be done in
time $O(n^3)$, and computation of $T \nabla SN$ also takes $O(n^3)$. The critical observation is that (by Observation
\ref{obs:triv}) every SN-set in $T$ appears exactly once as a leaf inside an execution of SL-$k$. The overall running
time is thus of the form $O( \sum_{i} (n^3 + s_{i}^{3k+3}))$ where $\sum s_i$ is equal to the total number of SN-sets
in $T$. Noting that $\sum_{i} s_i^{3k+3} \leq (\sum_{i} s_i)^{3k+3}$, and that there are at most $O(n)$ SN-sets in
$T$, we obtain for $k \geq 1$ an overall running time of $O( n^{3k+3})$, which is $O( |T|^{k+1})$ because $T$ is
dense. \qed
\end{proof}

\section{Conclusions and open questions}

In this article we have shown that, for level 1 and 2, constructing a phylogenetic network consistent with a dense set
of triplets that minimises the number of reticulations (i.e. DMRL-1 and DMRL-2), is polynomial-time solvable. We feel
that, given the widespread use of the principle of parsimony within phylogenetics, this is an important development,
and testing on simulated data has yielded promising results. However, the complexity of finding a \emph{feasible}
solution for level-3 and higher, let alone a minimum solution, remains unknown, and this obviously requires attention.
Perhaps the feasibility and minimisation variants diverge in complexity, for high enough $k$, it would be fascinating
to explore this.\\
\\
We have also shown, for \emph{every} fixed $k$, how to generate in polynomial time all simple level-$k$ networks consistent with
a dense set of triplets. This could be an important step towards determining whether the aforementioned feasibility question
is tractable for every fixed $k$. We have used the SL-$k$ algorithm to show how, for every fixed $k$, MIN-REFLECT-$k$ is
polynomial-time solvable. Clearly the demand that a set of triplets is exactly equal to the set of triplets in some network
is an extremely strong restriction on the input. However, for small networks and high accuracy triplets
such an assumption might indeed be valid, and thus of practical use. In any case, the concept of reflection is likely
to have a role in future work on ``support'' for edges in phylogenetic networks generated via triplets. Also, there remain some
fundamental questions to be answered about reflection. For example, the complexity of the question ``does \emph{any} network $N$
reflect $T$?'' remains unclear.

\section*{Acknowledgements}
We thank Judith Keijsper, Matthias Mnich and Leen Stougie for many helpful discussions during the writing of the
paper.

\end{document}